\documentclass[12pt,british,english]{article}
\usepackage{mathpazo}
\usepackage[T1]{fontenc}
\usepackage[latin9]{inputenc}
\usepackage[a4paper]{geometry}
\usepackage[active]{srcltx}
\usepackage{amsthm}
\usepackage{amsmath}
\usepackage{amssymb}

\makeatletter
\newcommand{\lyxaddress}[1]{
\par {\raggedright #1
\vspace{1.4em}
\noindent\par}
}
  \theoremstyle{definition}
  \newtheorem{defn}{\protect\definitionname}
  \theoremstyle{plain}
  \newtheorem{lem}{\protect\lemmaname}
\theoremstyle{plain}
\newtheorem{thm}{\protect\theoremname}
  \theoremstyle{plain}
  \newtheorem{conjecture}{\protect\conjecturename}
  \theoremstyle{plain}
  \newtheorem{cor}{\protect\corollaryname}

\usepackage{authblk}
\usepackage{babel}
\usepackage{amsfonts}
\usepackage{cite}

\usepackage[labelfont=bf, font=small]{caption}
\date{}

\newcommand{\kb}[2]{|#1\rangle\langle#2|}
\newcommand{\bk}[2]{\langle#1|#2\rangle}
\newcommand{\ket}[1]{|#1\rangle}
\newcommand{\bra}[1]{\langle#1|}

\usepackage{verbatim}
\usepackage[matrix,frame,arrow]{xy}
\xyoption{frame}
\xyoption{arrow}
\xyoption{arc}

\providecommand{\conjecturename}{Conjecture}
  \providecommand{\definitionname}{Definition}
  \providecommand{\lemmaname}{Lemma}
\providecommand{\corollaryname}{Corollary}
\providecommand{\theoremname}{Theorem}

\usepackage{babel}
\providecommand{\conjecturename}{Conjecture}
  \providecommand{\definitionname}{Definition}
  \providecommand{\lemmaname}{Lemma}
\providecommand{\corollaryname}{Corollary}
\providecommand{\theoremname}{Theorem}

\makeatother

\usepackage{babel}
  \addto\captionsbritish{\renewcommand{\conjecturename}{Conjecture}}
  \addto\captionsbritish{\renewcommand{\definitionname}{Definition}}
  \addto\captionsbritish{\renewcommand{\lemmaname}{Lemma}}
  \addto\captionsenglish{\renewcommand{\conjecturename}{Conjecture}}
  \addto\captionsenglish{\renewcommand{\definitionname}{Definition}}
  \addto\captionsenglish{\renewcommand{\lemmaname}{Lemma}}
  \providecommand{\conjecturename}{Conjecture}
  \providecommand{\definitionname}{Definition}
  \providecommand{\lemmaname}{Lemma}
\addto\captionsbritish{\renewcommand{\corollaryname}{Corollary}}
\addto\captionsbritish{\renewcommand{\theoremname}{Theorem}}
\addto\captionsenglish{\renewcommand{\corollaryname}{Corollary}}
\addto\captionsenglish{\renewcommand{\theoremname}{Theorem}}
\providecommand{\corollaryname}{Corollary}
\providecommand{\theoremname}{Theorem}

\begin{document}

\title{\textbf{\large{}{}Mutually Unbiased Product Bases for Multiple Qudits}{\large{}}}

\author{Daniel McNulty$^{1, 2}$, Bogdan Pammer$^{3}$ and Stefan Weigert$^{4}$}

\maketitle
\vspace{-1.4cm}
\begin{center}\date{17 March 2016}\end{center}

\lyxaddress{\begin{center}
\emph{\small{}{}$^{1}$Department of Optics, Palacký University,
17.~listopadu 12, 771~46 Olomouc, Czech Republic}{\small{}{}}\\
\emph{\small{}{}$^{2}$ Department of Applied Mathematics, Hanyang University (ERICA), 55 Hanyangdaehak-ro, Ansan, Gyeonggi-do, 426-791, Korea}{\small{}{}}\\
 \emph{\small{}{}$^{3}$Faculty of Physics, University of Vienna,
Boltzmanngasse 5, 1090 Vienna, Austria}{\small{}{}}\\
 \emph{\small{}{}$^{4}$Department of Mathematics, University
of York, York YO10 5DD, UK}{\small{} }
\par\end{center}}
\begin{abstract}
We investigate the interplay between mutual unbiasedness and product bases for multiple qudits of possibly different dimensions. A product state of such a system is shown to be mutually unbiased to a product basis only if each of its factors is mutually unbiased to all the states which occur in the corresponding factors of the product basis. This result implies both a tight limit on the number of mutually unbiased product bases which the system can support and a complete classification of mutually unbiased product bases for multiple qubits or qutrits. In addition, only maximally entangled states can be mutually unbiased to a maximal set of mutually unbiased product bases.
\end{abstract}
\global\long\def\kb#1#2{|#1\rangle\langle#2|}
 \global\long\def\bk#1#2{\langle#1|#2\rangle}
 \global\long\def\braket#1#2{\langle#1|#2\rangle}
 \global\long\def\ket#1{|#1\rangle}
 \global\long\def\bra#1{\langle#1|}
 \global\long\def\c#1{\mathbb{C}^{#1}}

\section{Introduction}

Complementarity is considered to be a fundamental concept of quantum
mechanics. Loosely speaking, two observables are complementary if
measuring one of them prevents an accurate simultaneous measurement
of the other. Position and momentum of a quantum particle, or two
spin components along different axes, are well-known examples. The
properties of complementary observables are crucial in the first
protocol of quantum key distribution \cite{bb84}.

Given a system residing in an eigenstate of one observable, the outcomes
of measuring a second observable are \emph{equally} likely if the
second observable is complementary to the first one. In other words,
the eigenbases of a complementary pair of observables are \emph{mutually
unbiased}. Explicitly, any two orthonormal bases $\{\ket{a_{i}}\}$
and $\{\ket{b_{j}}\}$ of dimension $d$ are mutually unbiased if and only if 
\begin{equation}
\left|\bk{a_{i}}{b_{j}}\right|^{2}=\frac{1}{d}\,,\qquad i,j=1\ldots d\,.
\end{equation}

For a qudit with Hilbert space of dimension $d$, the number of pairwise
MU bases is limited by $(d+1)$. The bound is tight \cite{ivanovic81,wootters89}
if the dimension equals the power of a prime number, $d=p^{n}$, $n\in\mathbb{N}$.
For other dimensions $d$, it is not known whether the maximum can
be reached. A proof is elusive even for the smallest case $d=6$,
although both rigorous results \cite{grassl04,jaming09,mcnulty2}
and substantial numerical evidence \cite{butterley07,brierley08}
support the conjectured maximum of \emph{three} MU bases.

In this paper we report a number of results which follow from the
assumption that the MU bases under consideration consist of \emph{product}
states only. Product bases do play an important role in the construction
of MU bases \cite{klappenecker04} if the dimension $d$ is not a
prime power. For instance, in bipartite dimensions $d=d_{1}d_{2}$,
MU bases can be built from the tensor products of sets of MU bases
in the subspaces $\c{d_{1}}$ and $\c{d_{2}}$. This construction
provides a lower bound on the number of MU bases in any composite
dimension, and it has been exceeded only in dimensions $d$ with specific
prime decompositions, using mutually orthogonal Latin squares \cite{wocjan05}.

Product bases also feature in complete sets of MU bases for prime
power dimensions $d=p^{n}$ since one can construct complete sets
of $(d+1)$ MU bases of which $(p+1)$ are product bases. Experimentally,
the distinction is important when implementing quantum information
tasks: product measurements on multiple qudits are easier to implement than entangled ones.

One of the main results of this paper is to show that, in a multipartite system,
the subsystem with the least number of MU bases severely restricts the
possibilities to construct MU product bases. This limitation allows
us to find a tight upper bound on the number of MU product bases for
multiple qudits, and to classify maximal sets of such bases.

The paper is set out as follows. In the next section we introduce
mutually unbiased product bases and prepare the ground by recalling
some results relevant in the present context. The third section contains
our first main result, a proof of a necessary and sufficient condition for
the construction of MU product bases in multipartite systems. In Sec.~\ref{sec:conjecture},
we derive a tight upper bound on the number of MU product bases in
a multipartite system with a subsystem of dimension two or three.
Applying this result, we then derive classifications of maximal sets
of MU product bases in a number of cases. In Sec.~\ref{sec:max_entangled_vectors}
we show that a vector mutually unbiased to a maximal set of MU product
bases must be maximally entangled with respect to a specific bipartition
of the system. A summary and some concluding remarks are presented
in Sec.~\ref{sec:conclusions}, along with a conjecture on the structure
of product bases in multipartite systems.

\section{Mutually unbiased product bases}

We start by defining product bases for a quantum system composed of
$n$ qudits, with dimension $d=d_{1}d_{2}\ldots d_{n}$. The state
space of the $r$-th qudit is the complex vector space $\mathbb{C}^{d_{r}}$,
with an integer $d_{r}\geq2$, $r=1\ldots n$. 
\begin{defn}
An orthonormal basis $\mathcal{B}$ of the complex vector space $\mathbb{C}^{d}$
with dimension $d=d_{1}d_{2}\ldots d_{n}$ is a \emph{product basis}
if each basis vector takes the form $\ket{\psi}=\ket{\psi^{1}}\otimes\ldots\otimes\ket{\psi^{n}}\in\mathbb{C}^{d}$,
with states $\ket{\psi^{r}}\in\mathbb{C}^{d_{r}}$, $r=1\ldots n$. 
\end{defn}
For a bipartite system with $d=d_{1}d_{2}$, two different types of
product bases exist, namely, \emph{direct} and \emph{indirect} product
bases, a distinction introduced in \cite{wiesniak11}. Direct product
bases consist of $d=d_{1}d_{2}$ states $\ket{v,V}\equiv\ket v\otimes\ket V$
where $\{\ket v,v=1\ldots d_{1}\}$ is an orthogonal basis of $\mathbb{C}^{d_{1}}$
and $\{\ket V,V=1\ldots d_2\}$ is an orthogonal basis of $\mathbb{C}^{d_{2}}$.

An important link between direct product bases and MU bases has been established
in \cite{wiesniak11}.
\begin{lem}
\label{lem:zeilinger}Two \emph{{[}}direct\emph{{]}} orthogonal product
bases $\{\ket{u,\mbox{U}}\}$ and $\{\ket{v,\mbox{V}}\}$ in dimension
$d=d_{1}d_{2}$ are \emph{MU} if and only if $\ket u$ is \emph{MU}
to $\ket v$ in dimension $d_{1}$ and $\ket V$ is \emph{MU} to $\ket{\mbox{U}}$
in dimension $d_{2}$. 
\end{lem}

Any orthogonal basis consisting of product states only -- but not of the form described by a direct product basis -- is called an indirect product basis. It may involve more than one orthogonal basis in one subsystem. For example, the set ${\cal B}=\{\ket{0,0},\ket{0,1},\ket{1,+},\ket{1,-}\}$
is an indirect product basis in dimension $d=2\times2$ since it contains
two different orthogonal bases in the second subsystem, namely, $\{\ket 0,\ket 1\}$
and $\{\ket +,\ket -\}$, with $\ket{\pm}=(\ket 0\pm\ket 1)/\sqrt{2}$.
In general, an indirect product basis of a bipartite system takes
the form $\mathcal{B}=\{\ket{\psi_{i}^{1},\psi_{i}^{2}},i=1\ldots d\}$,
with two sets $\{\ket{\psi_{i}^{1}}\in\mathbb{C}^{d_{1}}\}$ and $\{\ket{\psi_{i}^{2}}\in\mathbb{C}^{d_{2}}\}$ of $d$ states each.

It is important to recognize that indirect product bases are not equivalent to direct product bases under local unitary transformations. Therefore, the generalization of Lemma \ref{lem:zeilinger} to arbitrary pairs of product bases is not immediate.

Indirect product bases of systems with dimensions four and six were
investigated in \cite{mcnulty_all}, leading to a generalization of
Lemma \ref{lem:zeilinger}.
\begin{lem}
\label{lem:mcnulty/weigert} A product state $\ket{\mu^{1},\mu^{2}}\in\mathbb{C}^{d},d\equiv d_{1}d_{2}\leq6,$
is\emph{ MU} to the product basis $\{\ket{\psi_{i}^{1},\psi_{i}^{2}},i=1\ldots d\}$,
if and only if $\ket{\mu^{1}}$ is \emph{MU} to $\ket{\psi_{i}^{1}}\in\mathbb{C}^{d_{1}}$
and $\ket{\mu^{2}}$ is \emph{MU} to $\ket{\psi_{i}^{2}}\in\mathbb{C}^{d_{2}}$,
for all $i=1\ldots d$. 
\end{lem}
This result is strong enough to imply a classification of \emph{all}
MU product bases in dimensions four and six \cite{mcnulty_all}. It
turns out that there is only one way to construct three MU product
bases in the space $\mathbb{C}^{4}$ while two inequivalent MU product
triples exist in $\mathbb{C}^{6}$. For multipartite systems with
dimensions $d>6$, the set of inequivalent product bases
is not known. The proof of Lemma \ref{lem:mcnulty/weigert} relies
on exhaustively enumerating all (inequivalent, cf. below) product
bases in dimensions four and six. The following section presents an
alternative approach which allows us to generalize Lemma \ref{lem:mcnulty/weigert}
to arbitrary multipartite dimensions.

\section{Limiting the number of MU product vectors}

In this section, we generalize Lemma \ref{lem:mcnulty/weigert} to
multipartite systems of dimension $d=d_{1}d_{2}\ldots d_{n}$, with
$d_{r}\geq2$, $r=1\ldots n$, leading to Theorem \ref{thm:generalised_pammer}. The theorem will be important to
construct maximal sets of MU product bases.

In a first step, we generalize Lemma \ref{lem:mcnulty/weigert} to
arbitrary \emph{bipartite} systems with dimension $d=d_{1}d_{2}$ \cite{pammer15}.
\begin{lem}
\label{lem:pammer} A product state $\ket{\mu^{1},\mu^{2}}$ in dimension
$d=d_{1}d_{2}$ is \emph{MU} to any product basis $\mathcal{B}=\{\ket{\psi_{i}^{1},\psi_{i}^{2}},i=1\ldots d\}$
if and only if $\ket{\mu^{1}}$ is \emph{MU} to all states $\ket{\psi_{i}^{1}}\in\mathbb{C}^{d_{1}}$
and $\ket{\mu^{2}}$ is mutually unbiased to all states $\ket{\psi_{i}^{2}}\in\mathbb{C}^{d_{2}}$.\end{lem}
\begin{proof}
Assuming the relations $|\bk{\psi_{i}^{1}}{\mu^{1}}|^{2}=1/d_{1}$
and $|\bk{\psi_{i}^{2}}{\mu^{2}}|^{2}=1/d_{2}$, the state
$\ket{\mu^{1},\mu^{2}}$ is indeed found to be MU to the product states
of the basis $\mathcal{B}$,
\begin{equation}
\left|\bk{\psi_{i}^{1},\psi_{i}^{2}}{\mu^{1},\mu^{2}}\right|^{2}=\left|\bk{\psi_{i}^{1}}{\mu^{1}}\right|^{2}\left|\bk{\psi_{i}^{2}}{\mu^{2}}\right|^{2}=\frac{1}{d_{1}d_{2}}\,,\qquad i=1\ldots d\,.\label{eq: bipartite MU condition}
\end{equation}

To prove the converse we assume that $\ket{\mu^{1},\mu^{2}}$ is MU
to the states of the product basis $\mathcal{B}$, i.e. Eq. \eqref{eq: bipartite MU condition}. Let us now evaluate the traces of two projectors constructed from
the states $\ket{\mu^{1}}$ and $\ket{\mu^{2}}$, namely,
\begin{equation}\label{eq:trace1}
\begin{aligned}\mbox{tr }\left(\ket{\mu^{1}}\bra{\mu^{1}}\otimes\mathbb{1}\right) & =\sum_{i=1}^{d}\bra{\psi_{i}^{1},\psi_{i}^{2}}\left(\ket{\mu^{1}}\bra{\mu^{1}}\otimes\mathbb{1}\right)\ket{\psi_{i}^{1},\psi_{i}^{2}}\\
 & =\sum_{i=1}^{d}\left|\braket{\psi_{i}^{1}}{\mu^{1}}\right|^{2}=d_{2}\,,
\end{aligned}
\end{equation}
and
\begin{equation}\label{eq:trace2}
\mbox{tr }\left(\mathbb{1}\otimes\ket{\mu^{2}}\bra{\mu^{2}}\right)=\sum_{i=1}^{d}\left|\braket{\psi_{i}^{2}}{\mu^{2}}\right|^{2}=d_{1}\,.
\end{equation}
Defining the $2d$ positive numbers
\[
x_{i}=\sqrt{d_{1}}\left|\braket{\psi_{i}^{1}}{\mu^{1}}\right|\,,\qquad y_{i}=\sqrt{d_{2}}\left|\braket{\psi_{i}^{2}}{\mu^{2}}\right|\,,\qquad i=1\ldots d\,,
\]
the $(d+2)$ conditions (\ref{eq: bipartite MU condition})-(\ref{eq:trace2}) take the form 
\begin{equation}
x_{i}^{2}y_{i}^{2}=1\,,\qquad i=1\dots d\,,\label{eq: xy-product}
\end{equation}
and
\begin{equation}
\sum_{i=1}^{d}x_{i}^{2}=\sum_{i=1}^{d}y_{i}^{2}=d\,.
\end{equation}
These relations imply that
\begin{equation}
\sum_{i=1}^{d}\left(x_{i}-y_{i}\right)^{2}=0\,,
\end{equation}
which can only hold for 
\begin{equation}
x_{i}=y_{i}\,,\qquad i=1\ldots d\,.
\end{equation}
Using this result in Eq. \eqref{eq: xy-product} we see that indeed
\[
\left|\braket{\psi_{i}^{1}}{\mu^{1}}\right|^{2}=\frac{1}{d_{1}}\,,\qquad\left|\braket{\psi_{i}^{2}}{\mu^{2}}\right|^{2}=\frac{1}{d_{2}}\,,\qquad i=1\ldots d\,,
\]
must hold. Consequently, any state $\ket{\mu^{1},\mu^{2}}$ MU to
the states of a product basis ${\cal B}$ must have factors $\ket{\mu^{1}}$
and $\ket{\mu^{2}}$ which are MU to all states in the respective
subsystems. 
\end{proof}
The proof of Lemma \ref{lem:pammer} would be straightforward if we could transform any basis $\mathcal{B}$ to the canonical (direct) product basis by local unitary operations. However, this approach is not sufficiently general since no such transformations exist for indirect product bases $\mathcal{B}$.

Finally, we show that Lemma \ref{lem:pammer} can be generalized to obtain a result
about states MU to multi-partite orthogonal product bases. 
\begin{thm}
\label{thm:generalised_pammer} The product state $\ket{\mu^{1},\mu^{2},\ldots,\mu^{n}}$
in dimension $d=d_{1}d_{2}\ldots d_{n}$ is \emph{MU} to the orthogonal
product basis $\mathcal{B}=\{\ket{\psi_{i}^{1},\psi_{i}^{2},\ldots,\psi_{i}^{n}},i=1\ldots d\}$,
if and only if, for each $r=1\ldots n$, the state $\ket{\mu^{r}}$
is \emph{MU} to $\ket{\psi_{i}^{r}}\in\mathbb{C}^{d_{r}}$, for all
$i=1\ldots d$. \end{thm}
\begin{proof}
To derive this result, we consider the basis $\mathcal{B}$ as a bipartite
product basis $\{\ket{\psi_{i}^{r},\psi{}_{i}^{\overline{r}}}\}$
of the space $\mathbb{C}^{d}=\mathbb{C}^{d_{r}}\otimes\mathbb{C}^{d_{\overline{r}}}$
, where now $\ket{\psi_{i}^{r}}\in\mathbb{C}^{d_{r}}$, $\ket{\psi_{i}^{\overline{r}}}\in\mathbb{C}^{d_{\overline{r}}}$,
with $d_{\overline{r}}=d/d_{r}$ and $r=1\ldots n$. Similarly, the
state $\ket{\mu^{1},\mu^{2},\ldots,\mu^{n}}$ is written as $\ket{\mu^{r},\mu^{\overline{r}}}$
where $\ket{\mu^{r}}\in\mathbb{C}^{d_{r}}$ and $\ket{\mu^{\overline{r}}}\in\mathbb{C}^{d_{\overline{r}}}$.
Applying Lemma \ref{lem:pammer} to each of the $n$ bipartitions,
we conclude that $\ket{\mu{}^{r}}$ is MU to $\ket{\psi_{i}^{r}}\in\mathbb{C}^{d_{r}}$,
for all $r=1\ldots n$. 
\end{proof}

\section{Limiting the number of MU product bases}

\label{sec:conjecture}

In this section we present a tight upper bound on the number of MU
product bases in multipartite systems with $d=d_{1}d_{2}\ldots d_{n}$
whenever at least one subsystem (which we can choose to be the first one)
has a dimension smaller than four, i.e. $d_{1}=2$ or $d_{1}=3$. 
\begin{thm}
\label{maximal_sets_d=00003D00003D2} Suppose $d=d_{1}d_{2}\ldots d_{n}$,
and let $d_{1}=2$ or $d_{1}=3$, and $d_{1}\leq d_{r},r=2\ldots n$.
Then there exist at most $(d_{1}+1)$ \emph{MU} product bases in $\mathbb{C}^{d}$. 
\end{thm}

\begin{proof}This result follows if we can show that a product basis of the space
$\mathbb{C}^{d_{1}}\otimes\mathbb{C}^{d_{r}}$ with $d_{1}\leq d_{r}$
contains a subset of $d_{1}$ orthogonal states in the subspace $\mathbb{C}^{d_{1}}$
for $d_{1}=2$ or $d_{1}=3$. To draw this conclusion, we first prove
a lemma on the existence of orthogonal bases in the subsystems of a \emph{bipartite}
orthonormal product basis. 
\begin{lem}
\label{lem:orthog_tiple} Consider an orthogonal product basis $\mathcal{B}=\{\ket{a_{i},b_{i}},i=1\ldots d\}$
in dimension $d=d_{1}d_{2}$, with $d_{1}=2$ or $d_1=3$. Then, for every vector
$\ket{a_{\kappa},b_{\kappa}},\kappa\in\{1\ldots d\}$, there exists
a subset $\mathcal{B}_{\kappa}$ of ${\cal B}$, with elements $\{\ket{a_{\kappa},b_{\kappa}},\ket{a_{\lambda},b_{\lambda}},\ldots\}$,
such that the vectors $\{\ket{a_{\kappa}},\ket{a_{\lambda}},\ldots\}$
constitute an orthonormal basis of $\mathbb{C}^{d_{1}}$. 
\end{lem}
It will be useful to call two orthogonal product vectors of $\mathbb{C}^{d_{1}}\otimes\mathbb{C}^{d_{2}}$
$r$-\emph{ortho\-go\-nal}, with $r=1,2$, if the two vectors of
the $r$-th subsystem are orthogonal. For example, the state $\ket{1,+}$
of a qubit pair is $1$-orthogonal to $\ket{0,0}$ but not $2$-orthogonal,
while the state $\ket{1,1}$ is both $1$- and $2$-orthogonal to
$\ket{0,0}$. This concept extends naturally to $n$-partite systems.

To show Lemma \ref{lem:orthog_tiple} we proceed in two steps. To
begin, we show that for each vector $\ket{a_{\kappa},b_{\kappa}}$
of a given product basis of $\mathbb{C}^{d_{1}d_{2}}$ we can find
$(d_{1}-1)$ vectors which are $1$-orthogonal to $\ket{a_{\kappa},b_{\kappa}}$
but not $2$-orthogonal. We call this set ${\cal A}_{\kappa}$. For
$d_{1}=2$, this result already ensures that the basis vectors of
the first system contain an orthonormal basis in the $\mathbb{C}^{2}$
subsystem.

We then show that for any state $\ket{a_{\lambda},b_{\lambda}}\in\mathcal{A}_{\kappa}$,
there is a set $\mathcal{A}_{\kappa\lambda}\subset\mathcal{B}$ of
$(d_{1}-2)$ vectors which are $1$-orthogonal
but not $2$-orthogonal to $\ket{a_\kappa,b_\kappa}$ and $\ket{a_\lambda,b_\lambda}$. Consequently, the set $\mathcal{A}_{\kappa\lambda}$
contains one vector if $d_{1}=3$ which means that we have identified
\emph{three} orthogonal vectors in $\mathbb{C}^{3}$. 
\begin{proof}
\textbf{Step 1}: Choose any vector $\ket{a_{\kappa},b_{\kappa}}$
of the given orthonormal product basis ${\cal B}$. Then, each of
the remaining $(d-1)$ vectors is either $2$-orthogonal to it
or not. Let us partition the $(d-1)$ integers $i=1\ldots d$,
$i\neq\kappa$, accordingly into two sets,
\begin{align}
\mathcal{I}_{\kappa} & =\{i:\bk{b_{\kappa}}{b_{i}}\neq0,i\neq\kappa\}\,,\\
\mathcal{I}_{\overline{\kappa}} & =\{i:\bk{b_{\kappa}}{b_{i}}=0\}\,.
\end{align}
We denote the associated sets of states by
\begin{align}
\mathcal{A}_{\kappa} & =\{\ket{a_{i},b_{i}},i\in \mathcal{I}_{\kappa}\}\quad\mbox{and}\quad\mathcal{A}_{\overline{\kappa}}=\{\ket{a_{i},b_{i}},i\in \mathcal{I}_{\overline{\kappa}}\}\,,
\end{align}
respectively. Since the states in $\mathcal{A}_{\kappa}$ are not
$2$-orthogonal to $\ket{a_{\kappa},b_{\kappa}}$, they must be $1$-orthogonal,
i.e. the factors of the first subsystem satisfy the relation $\bk{a_{\kappa}}{a_{i}}=0$,
$i\in \mathcal{I}_{\kappa}$. Effectively, we have split the product basis of
$\mathbb{C}^{d}$ into three disjoint sets, 
\begin{equation}
\mathcal{B}=\{\ket{a_{\kappa},b_{\kappa}}\}\cup\mathcal{A}_{\kappa}\cup\mathcal{A}_{\overline{\kappa}}\,.
\end{equation}

To show that ${\cal A}_{\kappa}$ is not empty, we evaluate the trace
of the product $M\otimes\kb{b_{\kappa}}{b_{\kappa}}$ in two ways,
where $M$ is an arbitrary operator acting on $\mathbb{C}^{d_{1}}$.
We have, of course,
\begin{align}
\mbox{tr}\left(M\otimes\kb{b_{\kappa}}{b_{\kappa}}\right) & =\mbox{tr}_{1}M\,,\label{eq: tr M 1}
\end{align}
and, using the orthonormal product basis ${\cal B}$, we also find

\begin{equation}
\mbox{tr}\left(M\otimes\kb{b_{\kappa}}{b_{\kappa}}\right)=\sum_{i=1}^{d}M_{i}\left|\bk{b_{i}}{b_{\kappa}}\right|^{2}=M_{\kappa}+\sum_{i\in \mathcal{I}_{\kappa}}M_{i}|\bk{b_{i}}{b_{\kappa}}|^{2}\,,\label{eq: tr M 2}
\end{equation}
where $M_{i}\equiv\bra{a_{i}}M\ket{a_{i}}$. The expressions on the
right-hand side of the last two equations must coincide for any operator
$M$. This is only possible if the vector $\ket{a_{\kappa}}$ combined
with the states $\{\ket{a_{i}},i\in \mathcal{I}_{\kappa}\}$, i.e. those present
in the product states of $\mathcal{A}_{\kappa}$, span the space $\mathbb{C}^{d_{1}}$.
If they do not, define $M=\kb{\chi}{\chi}$, where $\ket{\chi}$ is
any state in the complement of their span. This choice leads to a
contradiction since the right-hand side of \eqref{eq: tr M 2} evaluates
to zero while that of \eqref{eq: tr M 1} can be non-zero. Thus, the set
$\mathcal{A}_{\kappa}$ must contain at least $(d_{1}-1)$ elements.
For $d_{1}=2$, this result proves Lemma \ref{lem:orthog_tiple}.

\textbf{Step 2}: If we now pick an arbitrary element $\ket{a_{\lambda},b_{\lambda}}$
from $\mathcal{A}_{\kappa}$ and apply a reasoning parallel to Step
1, the product basis can be further divided into the following disjoint
subsets 
\begin{equation}
\mathcal{B}=\{\ket{a_{\kappa},b_{\kappa}},\ket{a_{\lambda},b_{\lambda}}\}\cup\mathcal{A}_{\kappa\lambda}\cup\mathcal{A}_{\overline{\kappa\lambda}}\,,
\end{equation}
where the sets of integers 
\begin{align}
\mathcal{I}_{\kappa\lambda} & =\{i:\bk{b_{\kappa}}{b_{i}}\bk{b_{\lambda}}{b_{i}}\neq0,i\neq\kappa,i\neq\lambda\}\,\subset \mathcal{I}_{\kappa}\,,\\
\mathcal{I}_{\overline{\kappa\lambda}} & =\{i:\bk{b_{\kappa}}{b_{i}}\bk{b_{\lambda}}{b_{i}}=0\}\,,
\end{align}
give rise to the sets of states ${\cal A}_{\kappa\lambda}=\{\ket{a_i,b_i},i\in \mathcal{I}_{\kappa\lambda}\}$ and ${\cal A}_{\overline{\kappa\lambda}}=\{\ket{a_i,b_i},i\in \mathcal{I}_{\overline{\kappa\lambda}}\}$,
respectively. We want to show that the vectors of $\mathcal{A}_{\kappa\lambda}$
in conjunction with $\ket{a_{\lambda}}$ and $\ket{a_{\kappa}}$ form
an orthonormal basis of the first subsystem. To do so, we express
the trace over an arbitrary operator $M$ on $\mathbb{C}^{d_{1}}$
as 
\begin{equation}
\mbox{tr}_{1}M=\frac{1}{\bk{b_{\kappa}}{b_{\lambda}}}\mbox{tr}(M\otimes\ket{b_{\lambda}}\bra{b_{\kappa}})\,,\label{eq:trace-1}
\end{equation}
where the number $\bk{b_{\kappa}}{b_{\lambda}}$ is different from zero
since these vectors are not $2$-orthogonal by construction. Using
the product basis ${\cal B}$ of $\mathbb{C}^{d}$ to evaluate the
right-hand side of Eq. (\ref{eq:trace-1}), we find 
\begin{equation}
\begin{aligned}\frac{1}{\bk{b_{\kappa}}{b_{\lambda}}}\sum_{i=1}^{d}M_{i}\bk{b_{i}}{b_{\lambda}}\bk{b_{\kappa}}{b_{i}}=M_{\kappa}+M_{\lambda}+\frac{1}{\bk{b_{\kappa}}{b_{\lambda}}}\sum_{i\in \mathcal{I}_{\kappa\lambda}}M_{i}\bk{b_{i}}{b_{\lambda}}\bk{b_{\kappa}}{b_{i}}\,,\end{aligned}
\end{equation}
since the terms in the sum with labels from the set $\mathcal{I}_{\overline{\kappa\lambda}}$
do not contribute. The relation 
\begin{equation}
\mbox{tr}_{1}M=M_{\kappa}+M_{\lambda}+\sum_{i\in \mathcal{I}_{\kappa\lambda}}M_{i}\frac{\bk{b_{i}}{b_{\lambda}}\bk{b_{\kappa}}{b_{i}}}{\bk{b_{\kappa}}{b_{\lambda}}}
\end{equation}
must hold for all choices of $M$. In analogy to Step 1, the vectors
$\{\ket{a_{i}},i\in \mathcal{I}_{\kappa\lambda}\}$ must, when supplemented
with $\ket{a_{\kappa}}$ and $\ket{a_{\lambda}}$, span the space
$\mathbb{C}^{d_{1}}$ in order to avoid a contradiction. Hence, $\mathcal{A}_{\kappa\lambda}$
has at least $(d_{1}-2)$ elements. If $d_{1}=3$, we have shown the
existence of three orthogonal vectors in $\mathbb{C}^{d_{1}}$ which
completes the proof of Lemma \ref{lem:orthog_tiple}. \end{proof}

A product basis of $\c{d_{1}\ldots d_{n}}$ is also a product basis
of $\c{d_{1}}\otimes\c{d_{\overline{1}}}$, with $d_{\overline{1}}=d/d_{1}$.
Thus, Lemma \ref{lem:orthog_tiple} implies that each MU product basis
of $\mathbb{C}^{d}$ contains a subset of $d_{1}$ orthogonal states
in the subspace $\c{d_{1}}$ when $d_1=2$ or $d_1=3$. On the basis of Theorem \ref{thm:generalised_pammer}
we finally conclude that at most $(d_{1}+1)$ MU product bases exist
in $\c{d_{1}\ldots d_{n}}$ so that Theorem \ref{maximal_sets_d=00003D00003D2}
holds.
\end{proof}
The bound we obtain in Theorem \ref{maximal_sets_d=00003D00003D2} suggests
that a more general result holds. 
\begin{conjecture}
\label{conj:maximal_sets} Suppose $d=d_{1}d_{2}\ldots d_{n}$. Then
there exist at most $(d_{m}+1)$ \emph{MU} product bases in $\mathbb{C}^{d}$,
where $d_{m}$ is the dimension of the subsystem with the \emph{least} number of \emph{MU} bases.
\end{conjecture}
One way to prove the conjecture is to check whether the collection
of $d$ vectors figuring in the $m$-th subsystem of a product
basis of $\mathbb{C}^{d}$ contain an orthonormal basis of the space
$\mathbb{C}^{d_{m}}$. Assuming this to be true, Theorem
\ref{thm:generalised_pammer} limits the number of MU product bases
of $\mathbb{C}^{d}$ since only $(d_{m}+1)$ MU bases exist in $\mathbb{C}^{d_{m}}$.
The subsystem with the smallest number of MU bases, i.e. $\mathbb{C}^{d_m}$, therefore
restricts the number of MU product bases which can exist in the system
of dimension $d$ --- which is the content of Conjecture \ref{conj:maximal_sets}.
Our proof of Theorem \ref{maximal_sets_d=00003D00003D2} implements
exactly this strategy but we are not able to include higher dimensions.

\section{Maximal sets of MU product bases}

\label{sec:classification}

All sets of MU product bases are known for bipartite systems with
dimensions $d=2\times2$ and $d=2\times3$ \cite{mcnulty_all}. For $d=6$ they
contain several continuous families of MU product pairs and two triples.
Theorem \ref{maximal_sets_d=00003D00003D2} allows us to draw conclusions
about the structure of MU product bases of more general bipartite
and multipartite systems with dimensions $d=d_{1}\ldots d_{n}$,
as long as $d_{1}=2$ or $d_{1}=3$. In particular, we will enumerate
\emph{maximal} sets of $(d_{1}+1)$ MU product bases if $d=2^{n}$
and $d=3^{n}$, identifying a unique triple and quadruple of MU product
bases, respectively. \emph{Inequivalent} triples and quadruples, respectively,
are found to exist already for $d=2p$ and $d=3p$, with prime numbers
$p\geq5$.

MU product bases are \emph{equivalent }if they can be mapped onto
each other without affecting both the product structure of the states
and the modulus of their inner products. As explained in \cite{mcnulty_all},
the allowed equivalence transformations consist of local unitary maps
acting on all bases simultaneously, the multiplication of any state
by an arbitrary phase factor, the permutation of states within a basis,
and the local complex conjugation of all bases; in addition, the bases
may be written down in an arbitrary order.

To begin, we recall the unique complete sets of MU bases of the spaces
$\mathbb{C}^{2}$ and $\mathbb{C}^{3}$, expressing each basis as
a square matrix, with columns given by the components of (unnormalised)
basis vectors relative to the standard basis. In dimension $d=2$,
one triple of MU product bases exists, 
\begin{equation}
\left(\begin{array}{rr}
1 & 0\\
0 & 1
\end{array}\right)\,,\quad\left(\begin{array}{rr}
1 & 1\\
1 & -1
\end{array}\right)\,,\quad\left(\begin{array}{rr}
1 & 1\\
i & -i
\end{array}\right)\,,\label{eq: complete set d=00003D2}
\end{equation}
where any other triple associated with the space $\mathbb{C}^{2}$ is equivalent
to this one. Clearly, the bases are determined by the eigenstates
of the Pauli operators $\sigma_{z}$, $\sigma_{x}$ and $\sigma_{y}$,
respectively, $\{\ket{j_{z}}\}$, $\{\ket{j_{x}}\}$ and $\{\ket{j_{y}}\}$,
with $j=0,1$.

For $d=3$, the set of four MU bases 
\begin{equation}
\left(\begin{array}{ccc}
1 & 0 & 0\\
0 & 1 & 0\\
0 & 0 & 1
\end{array}\right),\quad\left(\begin{array}{ccc}
1 & 1 & 1\\
1 & \omega & \omega^{2}\\
1 & \omega^{2} & \omega
\end{array}\right),\quad\left(\begin{array}{ccc}
1 & 1 & 1\\
\omega & \omega^{2} & 1\\
\omega & 1 & \omega^{2}
\end{array}\right),\quad\left(\begin{array}{ccc}
1 & 1 & 1\\
\omega^{2} & 1 & \omega\\
\omega^{2} & \omega & 1
\end{array}\right)\,,\label{complete set d=00003D3}
\end{equation}
where $\omega=e^{2\pi i/3}$, is also unique up to equivalence. Here
the column vectors emerge as the eigenstates $\{\ket{J_{z}}\}$, $\{\ket{J_{x}}\}$,
$\{\ket{J_{y}}\}$ and $\{\ket{J_{w}}\}$, $J=0,1,2$, of generalized
Pauli operators $Z,X,XZ$, and $XZ^2$ in $\mathbb{C}^{3}$. The
operators give rise to the discrete Heisenberg-Weyl group in $\mathbb{C}^{3}$,
via the relation $ZX=\omega XZ$, where $X$ and $Z$ are the Heisenberg-Weyl shift and phase operators, respectively.

\subsubsection*{Dimension $d=2^{n}$}
Here we consider product bases of dimension $d=d_1\ldots d_n$, with $d_r=2$, for each $r=1\ldots n$.
\begin{cor}
\label{maximal_sets_p=00003D00003D2} In the space $\mathbb{C}^{d}$
with dimension $d=2^{n}$, a unique triple of \emph{MU} product bases
exists, 
\begin{align}
\mathcal{B}_{0} & =\{\ket{j_{z}^{1}}\otimes\ldots\otimes\ket{j_{z}^{n}}\},\label{b0}\\
\mathcal{B}_{1} & =\{\ket{j_{x}^{1}}\otimes\ldots\otimes\ket{j_{x}^{n}}\},\label{b1}\\
\mathcal{B}_{2} & =\{\ket{j_{y}^{1}}\otimes\ldots\otimes\ket{j_{y}^{n}}\},\label{b2}
\end{align}
up to local equivalence transformations; here $\{\ket{j_{b}^{r}}, j=0,1\}$, $b=z,x,y$, are, for each $r=1\ldots n$, the eigenstates of the three Pauli operators in $\mathbb{C}^{2}$. 
\end{cor}
\begin{proof}First we show that any triple of MU product bases in
dimension $d=2^{n}$ consists of direct product bases. According to
Lemma \ref{lem:orthog_tiple}, every product basis of dimension $d=2q$
contains a pair of orthogonal states in the $\mathbb{C}^{2}$ subspace.
Hence, Theorem \ref{thm:generalised_pammer} implies that each product
basis of an MU triple in dimension $d=2q$ contains a unique pair
of orthogonal states in $\mathbb{C}^{2}$. Applying this argument
to all bipartitions $\mathbb{C}^{2}\otimes\mathbb{C}^{2^{n-1}}$ of
$\mathbb{C}^{2^{n}}$, we conclude that only one pair of orthogonal
states occurs in each subsystem $\mathbb{C}^{2}$. Therefore, all three
MU bases are direct product bases.

By performing local unitary transformations, we turn the first basis
into the standard basis, displayed in Eq. (\ref{b0}). The remaining
two bases $\mathcal{B}_{1}$ and $\mathcal{B}_{2}$ contain either
eigenstates of $\sigma_{x}$ or $\sigma_{y}$ in each of their subsystems.
Whenever the states $\ket{j_{y}}$ appear in a subsystem of the second
basis we apply the unitary transformation $\kb{0_{z}}{0_{z}}+\omega\kb{1_{z}}{1_{z}}$,
with $\omega=i$. This operation, a rotation by $\pi/2$ about
the $z$-axis exchanges the operators $\sigma_{x}$ and $\sigma_{y}$,
hence their eigenstates, and it leaves the states $\ket{j_{z}}$ unchanged,
up to phase factors. These properties of the transformation can be
directly verified by inspecting Eq. (\ref{eq: complete set d=00003D2}).
Hence, the second and third bases can always be mapped to $\mathcal{B}_{1}$
and $\mathcal{B}_{2}$, respectively, which completes the proof of
Corollary \ref{maximal_sets_p=00003D00003D2}.\end{proof}

\subsubsection*{Dimension $d=3^{n}$}
We now prove an analogous result for product bases of dimension $d=d_1\ldots d_n$, where $d_r=3$, for each $r=1\ldots n$.
\begin{cor}
\label{maximal_sets_p=00003D00003D3} In the space $\mathbb{C}^{d}$
with dimension $d=3^{n}$, a unique quadruple of \emph{MU} product
bases exists, 
\begin{align}
\mathcal{B}_{0} & =\{\ket{J_{z}^{1}}\otimes\ldots\otimes\ket{J_{z}^{n}}\},\label{b0 d=00003D3^n}\\
\mathcal{B}_{1} & =\{\ket{J_{x}^{1}}\otimes\ldots\otimes\ket{J_{x}^{n}}\},\label{b1 d=00003D3^n}\\
\mathcal{B}_{2} & =\{\ket{J_{y}^{1}}\otimes\ldots\otimes\ket{J_{y}^{n}}\},\label{b2 d=00003D3^n}\\
\mathcal{B}_{3} & =\{\ket{J_{w}^{1}}\otimes\ldots\otimes\ket{J_{w}^{n}}\},\label{b3 d=00003D3^n}
\end{align}
up to local equivalence transformations; here\emph{ }$\{\ket{J_{b}^{r}}, J=0,1,2\}$,
$b=z,x,y,w$, are, for each $r=1\ldots n$, the eigenstates of
$Z$, $X$, $XZ$ and ${XZ^{2}}$, respectively, where $X$ and $Z$ are the Heisenberg-Weyl shift and phase operators in $\mathbb{C}^3$.
\end{cor}
\begin{proof}
As in the previous case, we first use Lemma \ref{lem:orthog_tiple}
and Theorem \ref{thm:generalised_pammer} to conclude that all four
MU product bases must be \emph{direct} product bases, constructed
from various tensor products of the complete set of four MU bases
in dimension three. Any one of these product bases can always be transformed
to the standard basis $\mathcal{B}_{0}$ by a suitable product of
local unitary operations. Then, the remaining three product bases consist
of tensor products of various combinations of the bases $\{\ket{J_{b}^{r}},J=0,1,2\}$, $b=w,x,y$. Pick any of these three bases and apply the operator
$\kb{0_{z}}{0_{z}}+\omega^{k}\kb{1_{z}}{1_{z}}+\omega^{k}\kb{2_{z}}{2_{z}}$,
where $\omega=e^{2\pi i/3}$ and $k\in\{0,1,2\}$, in the following
way: choose $k=0$ (or $k=1$ or $k=2$) for the factors which contain
the $x$-basis (or the $y$- basis or the $w$-basis, respectively).
This operation maps the product bases of the second basis to the tensor
product Fourier basis $\mathcal{B}_{1}$ which can be seen directly
upon inspecting the expressions given in \eqref{complete set d=00003D3}.
The states of the basis ${\cal B}_{0}$ only pick up irrelevant phase
factors during this process. 

Finally, the last two bases must be products of either $\{\ket{J_{y}}\}$
and $\{\ket{J_{w}}\}$. Picking one of them, each factor $\{\ket{J_{w}}\}$
can be turned into $\{\ket{J_{y}}\}$ by a local complex conjugation
which swaps $\{\ket{J_{y}}\}$ and $\{\ket{J_{w}}\}$ and leaves invariant
the bases $\{\ket{J_{z}}\}$ and $\{\ket{J_{y}}\}$ (see \eqref{complete set d=00003D3}).
Therefore, the last two bases have indeed been mapped to the product
bases $\mathcal{B}_{2}$ and $\mathcal{B}_{3}$ listed in Corollary
\ref{maximal_sets_p=00003D00003D3}.
\end{proof}

\subsubsection*{Dimension $d=2\times5$}

In dimension five, there exists a single complete set of six MU bases.
We shall denote these bases by $\mathcal{G}_{i}$, $i=0\ldots5$,
and refer to \cite{brierley_all} for their explicit form.

Theorem \ref{maximal_sets_d=00003D00003D2} implies that at most three
MU product bases exist in dimension $d=2\times5$. In addition, Theorem
\ref{thm:generalised_pammer} has implications for their structure. 
\begin{cor}
\label{maximal_sets_p=00003D00003D2x5} In the space $\mathbb{C}^{d}$
with dimension $d=2\times5$, any triple of \emph{MU} product bases
must be of the form 
\begin{align}
\mathcal{B}_{0} & =\{\ket{0_{z}}\otimes\mathcal{G}(0_z)\,,\,\ket{1_{z}}\otimes\mathcal{G}({1_z})\},\label{2x5triple1}\\
\mathcal{B}_{1} & =\{\ket{0_{x}}\otimes\mathcal{G}({0_x})\,,\,\ket{1_{x}}\otimes\mathcal{G}({1_x})\},\label{2x5 triple2}\\
\mathcal{B}_{2} & =\{\ket{0_{y}}\otimes\mathcal{G}({0_y})\,,\,\ket{1_{y}}\otimes\mathcal{G}({1_y})\},\label{2x5triple3}
\end{align}
up to local equivalence transformations; here $\{\ket{j_{b}}, j=0,1\}$, $b=z,x,y$, are the eigenstates of the three Pauli operators in $\mathbb{C}^{2}$, and $\mathcal{G}({j_b})$ are bases of $\c 5$ for each $j_b$, such that $\mathcal{G}({j_b})\in\mathcal{G}_i$. \end{cor}
\begin{proof}
As we have already seen, every product basis of dimension $d=2q$
contains a pair of orthogonal states in the $\c 2$ subspace. For
a set of three MU product bases, each basis contains one unique pair,
given by the eigenstates of $\sigma_{x}$, $\sigma_{y}$ and $\sigma_{z}$.
To satisfy orthogonality, each state in $\c 2$ is paired with an
orthogonal basis in $\c 5$. These six bases in $\c 5$, according
to Theorem \ref{thm:generalised_pammer}, are grouped into three mutually
unbiased sets, $\{\mathcal{G}({0_z}),\mathcal{G}({1_z})\}$, $\{\mathcal{G}({0_x}),\mathcal{G}({1_x})\}$
and $\{\mathcal{G}({0_y}),\mathcal{G}({1_y})\}$. The bases within each
set are taken from the complete set of six MU bases $\mathcal{G}_{i}$,
$i=0\ldots5$. This follows from the fact that all inequivalent triples,
quadruples, quintuples and sextuples of MU bases in $\c 5$ are given
by subsets of the complete set \cite{brierley_all}. 
\end{proof}
Corollary \ref{maximal_sets_p=00003D00003D2x5} implies that several \emph{inequivalent} triples
of MU product bases exist. For example, the six bases in $\c 5$ may
be chosen such that $\mathcal{G}({0_z})=\mathcal{G}({1_z})$, $\mathcal{G}({0_x})=\mathcal{G}({1_x})$
and $\mathcal{G}({0_y})=\mathcal{G}({1_y})$, in which case $\mathcal{B}_{0}$,
$\mathcal{B}_{1}$, and $\mathcal{B}_{2}$ form direct product bases.
Alternatively, if none of the six bases coincide, three \emph{indirect}
MU product bases emerge.


\subsubsection*{Dimension $d=2^{k}d_{2}\ldots d_{n},k\in\mathbb{N}$}

Suppose we consider product bases of dimension $d=2^{k}d_{2}\ldots d_{n},k\in\mathbb{N}$,
in the space $({\c 2})^{\otimes k}\otimes\c{d_{2}}\otimes\ldots\otimes\c{d_{n}}$. We can generalize Corollary \ref{maximal_sets_p=00003D00003D2x5} as follows.
\begin{cor}
\label{maximal_sets_d=00003D2^kxD} In the space $\mathbb{C}^{d}$
with dimension $d=2^{k}d_{2}d_{3}\ldots d_{n}$, any triple of \emph{MU}
product bases must be of the form 
\begin{align}
\mathcal{B}_{0} & =\{\ket{j_{z}}\otimes\mathcal{G}({j_z})\},\label{2^kxD triple1}\\
\mathcal{B}_{1} & =\{\ket{j_{x}}\otimes\mathcal{G}({j_x})\},\label{2^kxD triple2}\\
\mathcal{B}_{2} & =\{\ket{j_{y}}\otimes\mathcal{G}({j_y})\},\label{2^kxD triple3}\end{align}
up to local equivalence transformations; here $\{\ket{j_{b}}, j=0\ldots (2^k-1)\}$, $b=z,x,y$, are the eigenstates of $\sigma_{z}^{\otimes k}$, $\sigma_{x}^{\otimes k}$ and $\sigma_{y}^{\otimes k}$, respectively, and $\mathcal{G}({j_b})$ are bases of $\c {d_2}\otimes\ldots\otimes \c {d_n}$ for each $j_b$, such that the three sets $\{\mathcal{G}({j_b}), j=0\ldots (2^k-1)\}$ are mutually unbiased. 
\end{cor}
Note that the three sets $\{\mathcal{G}({j_b}), j=0\ldots (2^k-1)\}$, $b=z,x,y$, are mutually unbiased if the bases within each set are mutually unbiased to the bases of the other two sets. The bases within each set need not be mutually unbiased.

\subsubsection*{Dimension $d=3^{k}d_{2}\ldots d_{n},k\in\mathbb{N}$}
By considering product bases of the space $(\mathbb{C}^3)^{\otimes k}\otimes \mathbb{C}^{d_2}\otimes\ldots\otimes\mathbb{C}^{d_n}$ we find that any set of four MU product bases has the following structure.
\begin{cor}
\label{maximal_sets_d=00003D3^kxD} In the space $\mathbb{C}^{d}$
with dimension $d=3^{k}d_{2}d_{3}\ldots d_{n}$, any quadruple of \emph{MU}
product bases must be of the form 
\begin{align}
\mathcal{B}_{0} & =\{\ket{J_{z}}\otimes\mathcal{G}({J_z})\},\label{3^kxD triple1}\\
\mathcal{B}_{1} & =\{\ket{J_{x}}\otimes\mathcal{G}({J_x})\},\label{3^kxD triple2}\\
\mathcal{B}_{2} & =\{\ket{J_{y}}\otimes\mathcal{G}({J_y})\},\label{3^kxD triple3}\\
\mathcal{B}_{3} & =\{\ket{J_{w}}\otimes\mathcal{G}({J_w})\},\label{3^kxD triple4}
\end{align}
up to local equivalence transformations; here $\{\ket{J_{b}}, J=0\ldots(3^{k}-1)\}$,
$b=z,x,y,w$, are the eigenstates of
$Z^{\otimes k}$, $X^{\otimes k}$, $XZ^{\otimes k}$ and ${XZ^{2}}^{\otimes k}$,
respectively, where $X$ and $Z$ are the Heisenberg-Weyl shift and phase operators in $\mathbb{C}^3$, and $\mathcal{G}({J_b})$ are bases of $\c {d_2}\otimes\ldots\otimes \c {d_n}$ for each $J_b$, such that the four sets $\{\mathcal{G}({J_b}), J=0\ldots(3^{k}-1)\}$ are mutually unbiased.
\end{cor}
We omit the proofs of Corollaries \ref{maximal_sets_d=00003D2^kxD} and \ref{maximal_sets_d=00003D3^kxD} since they closely follow the proof of Corollary \ref{maximal_sets_p=00003D00003D2x5}.

\section{Vectors mutually unbiased to MU product bases}

\label{sec:max_entangled_vectors}

In a bipartite system with dimension $d=pq$, complete sets of MU
bases come with a fixed amount of entanglement \cite{wiesniak11}
which implies an upper bound on the number of MU \emph{product} bases
in a complete set: the space $\mathbb{C}^{pq}$ can accommodate at
most $(p+1)$ MU product bases for any pair of prime numbers
satisfying $p\leq q$. In addition, all of the remaining states must
be maximally entangled. If, for example, a hypothetical complete set
in dimension $d=2\times3$ contained \emph{three} MU product bases,
the other four bases would be maximally entangled.

Furthermore, it has been shown for $d=6$ that any vector MU to a
set of three MU product bases is maximally entangled \cite{mcnulty+12d}.
We will now generalize this property: a vector $\ket{\mu}\in\mathbb{C}^{d}$
of an $n$-partite qudit system with $d=d_{1}d_{2}\ldots d_{n}$ is
mutually unbiased to a set of $(d_{1}+1)$ MU product bases only if
$\ket{\mu}$ is maximally entangled. 
\begin{lem}
\label{lem:maximally_entangled} Let $d=d_{1}\ldots d_{n}$ with $d_{r}=p_{r}^{k_{r}}$,
$p_{r}$ prime and $k_{r}\in\mathbb{N}$, $r=1\ldots n$, such that
$d_{1}\leq\ldots\leq d_{n}$. A vector $\ket{\mu}$, mutually unbiased
to a set of $(d_{1}+1)$ \emph{MU }product bases \emph{(}where the product bases of $\mathbb{C}^d$ contain at least one orthogonal set of $d_1$ vectors in the subsystem $\c {d_1}$\emph{)}, is maximally entangled
across $\c{d_{1}}\otimes\c{d_{\overline{1}}}$, with $d_{\overline{1}}=d/d_{1}$.\end{lem}
\begin{proof}
Let us consider the $n$-partite system as a bipartite system with
state space $\mathbb{C}^{d_{1}}\otimes\mathbb{C}^{d_{\overline{1}}}$,
where $d_{\overline{1}}=d/d_{1}$. Following from Theorem \ref{thm:generalised_pammer} we write the set of $(d_{1}+1)$
MU product bases as $\mathcal{B}_{b}=\{\ket{v_{b},\overline{v}(v_{b})}, v=1\ldots d_{1}, \overline{v}=1\ldots d_{\overline{1}}\}$,
with $b=0\dots d_{1}$, such that $\{\ket{v_{b}}\}$ is an orthonormal basis of $\mathbb{C}^{d_{1}}$ for each $b$ and $\{\ket{\overline{v}(v_{b})}\}$ is an orthonormal basis of $\mathbb{C}^{d_{\overline{1}}}$ for each $v$ and $b$.

The unit vector $\ket{\mu}$ is MU to the product bases if the $d(d_{1}+1)$
equations 
\begin{equation}
|\bk{v_{b},\overline{v}(v_{b})}{\mu}|^{2}=\frac{1}{d}\,,
\end{equation}
are satisfied. Summing over all values of $\overline{v}$, we find
\begin{equation}
\sum_{\overline{v}=1}^{d_{\overline{1}}}|\bk{v_{b},\overline{v}(v_{b})}{\mu}|^{2}=\bra{v_{b}}(\text{tr}_{\overline 1}\kb{\mu}{\mu})\ket{v_{b}}=\bra{v_{b}}\rho_{1}\ket{v_{b}}=\frac{1}{d_{1}}\,,\label{eq:sum_over_v}
\end{equation}
where $\rho_{1}=\text{tr}_{\overline 1}\kb{\mu}{\mu}$ is the reduced density
matrix of the first subsystem, given by the partial trace of $\kb{\mu}{\mu}$
over the second subsystem.

We now show that Eqs. (\ref{eq:sum_over_v}) can only hold if the
state $\rho_{1}$ is maximally mixed. To see this, we rewrite $\rho_{1}$
in terms of a complete set of MU bases \cite{ivanovic81}, i.e. 
\begin{equation}
\rho_{1}=\sum_{b=0}^{d_{1}}\sum_{v=1}^{d_{1}}p_{b}^{v}\ket{v_{b}}\bra{v_{b}}-\mathbb{1},\label{eq:density_matrix}
\end{equation}
where $p_{b}^{v}\equiv\bra{v_{b}}\rho_{1}\ket{v_{b}}=1/d_{1}$ for
all $v$ and $b$, according to Eq. \eqref{eq:sum_over_v}. Using
$\sum_{v=1}^{d_{1}}p_{b}^{v}\kb{v_{b}}{v_{b}}=\mathbb{1}/d_{1}$ for
each basis, we find that Eq. (\ref{eq:density_matrix}) reduces to
\begin{equation}
\rho_{1}=\frac{1}{d_{1}}\mbox{\ensuremath{\mathbb{1}}}\,,
\end{equation}
which means that the state $\rho_{1}$ is maximally mixed, completing
the proof of Lemma \ref{lem:maximally_entangled}.
\end{proof}
For the special case of $d=p^{n}$, a stronger restriction on the
set of mutually unbiased vectors can be found. 
\begin{lem}
\label{lem:maximally_entangled_prime_power} Let $d=d_{1}\ldots d_{n}=p^{n}$
with $d_{r}=p$, $r=1\ldots n$, and a prime number $p$. A vector
$\ket{\mu}$, mutually unbiased to a set of $(p+1)$ \emph{MU} product bases
\emph{(}where the product bases of $\mathbb{C}^d$ contain at least one orthogonal set of $d_r$ vectors in each subsystem $\c {d_r}$\emph{)},
is maximally entangled across all bipartitions $\mathbb{C}^{p}\otimes\mathbb{C}^{p^{n-1}}$.\end{lem}
\begin{proof}
To show that $\ket{\mu}$ is maximally entangled, we apply Lemma \ref{lem:maximally_entangled}
to each of the $n$ possible bipartition $\mathbb{C}^{p}\otimes\mathbb{C}^{p^{n-1}}$.
Hence, the state $\ket{\mu}$ is maximally entangled across all such
bipartitions.
\end{proof}
It is interesting to compare the content of Lemma \ref{lem:maximally_entangled}
with results known for the cases $d=2d_{2}$ and $d=3d_{2}$. To do
so we adapt Lemma \ref{lem:maximally_entangled} accordingly.
\begin{cor}
\label{cor:maximally_entangled} Suppose that $d=d_{1}d_{2}$ with
$d_{1}=2$ or $d_{1}=3$, $d_{2}$ prime, and $d_{2}\geq d_{1}$.
Any vector $\ket{\mu}$, mutually unbiased to a set of $(d_{1}+1)$
\emph{MU} product bases of dimension $d$, is maximally entangled. 
\end{cor}
This statement is stronger than the one given in \cite{wiesniak11}
which states that, given a hypothetical complete set of $(d+1)$ MU
bases in dimension $d=d_{1}d_{2}$ containing $(d_{1}+1)$ MU product
bases, the remaining vectors must be maximally entangled. Corollary \ref{cor:maximally_entangled} is valid without assuming the existence of a complete set.

For dimension $d=6$, Corollary 6 implies that no vector is mutually unbiased to a set of three MU product bases \cite{mcnulty+12d}. We expect similar results to hold for larger product dimensions such as $d=2\times 5$, but we have not been able to generalize the proof for $d=6$.

Finally, let us make explicit Lemma \ref{lem:maximally_entangled_prime_power}
for the case of $n$ qubits or qutrits, i.e. $d=p^{n}$, with $p=2$
or $p=3$.
\begin{cor}
\label{cor:maximally_entangled_d=00003D00003D2n} Any vector $\ket{\mu}$,
mutually unbiased to a set of $(p+1)$ \emph{MU} product bases in
dimension $d=p^{n}$, with $p=2$ or $p=3$, is maximally entangled
with respect to every partition $\mathbb{C}^{p}\otimes\mathbb{C}^{p^{n-1}}$. 
\end{cor}

\section{Conclusions}

\label{sec:conclusions}

In this paper we investigated the relationship between product bases
and mutually unbiased bases for multipartite systems. Our first main
result is Theorem \ref{thm:generalised_pammer} which states that,
for \emph{any} dimension $d=d_{1}\ldots d_{n}$, a product vector
$\ket{\mu}$ is mutually unbiased to a product basis if and only if
the $r$-th factor of $\ket{\mu}$ is mutually unbiased to the $r$-th factor of each vector
present in the basis. This result considerably
generalizes what had been known before, for bipartite systems with
dimension four or six \cite{mcnulty_all}.

We also derived a tight upper bound on the number of MU product bases
in any composite dimension if at least one subsystem has dimension
two or three (Theorem \ref{maximal_sets_d=00003D00003D2}). We expect a similar bound
to hold in general, i.e. for \emph{all} composite dimensions, as described in Conjecture \ref{conj:maximal_sets}. One way to prove the conjecture
would be to show that a product basis of dimension $d=d_{1}d_{2}\ldots d_{n}$
contains an orthonormal set of $d_{r}$ states in the subspace $\mathbb{C}^{d_{r}}$,
for all $r=1\ldots n$ --- which we consider highly plausible.

Theorem \ref{thm:generalised_pammer} and Lemma \ref{lem:orthog_tiple}
allow us to classify all maximal sets of MU product bases in dimensions
$d=2^{n}$ and $d=3^{n}$. Somewhat surprisingly, only one triple
of MU product bases exists in dimension $d=2^{n}$ according to Corollary
\ref{maximal_sets_p=00003D00003D2}, and only one quadruple exists
for $d=3^{n}$ (Corollary \ref{maximal_sets_p=00003D00003D3}). Furthermore,
we have shown that \emph{inequivalent} triples of MU product bases
exist if $d=2\times5$, complementing a result of \cite{mcnulty_all}
which finds two such triples if $d=2\times3$.

Finally, we analysed the entanglement structure of vectors mutually
unbiased to product bases. We find that vectors mutually unbiased
to maximal sets of MU product bases must be maximally entangled (Lemmas
\ref{lem:maximally_entangled} and \ref{lem:maximally_entangled_prime_power}).
If one of the subsystems has dimension two or three, this result generalizes
to \emph{all} maximal sets of MU product bases (Corollaries \ref{cor:maximally_entangled}
and \ref{cor:maximally_entangled_d=00003D00003D2n}). This fact is
in line with the bipartite case $d=2\times3$ for which any vector
mutually unbiased to a set of three MU product bases had been shown
to be maximally entangled \cite{mcnulty2}.

We conclude by noting that all the evidence available to us points
to a natural and beautiful structure of orthogonal product bases in
multipartite quantum systems. For simplicity, we restrict ourselves
to the bipartite case.
\begin{conjecture}
\label{conj:prod_bases_structure} The set ${\cal B}=\left\{ \ket{a_{i},b_{i}},i=1\ldots d\right\} $
is an orthonormal product basis of the space $\c d$, with $d=d_{1}d_{2}$,
if and only if the $d$ vectors $\left\{ \ket{a_{i}}\in\c{d_{1}},i=1\ldots d\right\} $
and the $d$ vectors $\left\{ \ket{b_{i}}\in\c{d_{2}},i=1\ldots d\right\} $
can be grouped into $d_{2}$ orthonormal bases ${\cal B}_{i_{2}}(d_{1}),i_{2}=1\ldots d_{2}$,
and $d_{1}$ orthonormal bases ${\cal B}_{i_{1}}(d_{2}),i_{1}=1\ldots d_{1}$,
respectively. 
\end{conjecture}
Future progress towards a solution of the existence problem of MU
bases in non-prime power dimensions might take a twisted route involving
mutually unbiased product bases.

\subsubsection*{Acknowledgments}

D. M. acknowledges the support of the Operational Program Education
for Competitiveness Project No. CZ.1.07/2.3.00/30.0041 co-financed
by the European Social Fund and the Czech Ministry of Education. D. M. also acknowledges support by the Institute for Information \& communications Technology Promotion(IITP) grant funded by the Korea government(MSIP) (No.R0190-15-2028, PSQKD). B. P. would like to thank Anton Zeilinger for helpful discussions.


\begin{thebibliography}{10}

\bibitem{bb84} C. H. Bennett and G. Brassard, in \emph{Proceedings of IEEE International Conference on Computers, Systems and Signal Processing} \textbf{175}, 8 (1984).

\bibitem{ivanovic81} I. D. Ivanovi\'{c}, \textit{J. Phys. A} \textbf{14},
3241 (1981).

\bibitem{wootters89} W. K. Wootters and B. D. Fields, \textit{Ann.
Phys. (N.Y.)} \textbf{191}, 363 (1989).

\bibitem{grassl04} M. Grassl, in \textit{Proc. ERATO Conf. on Quantum
Information Science 2004 (EQIS 2004)}; also: e-print arXiv:quant-ph/0406175
(2004).

\bibitem{jaming09} P. Jaming, M. Matolcsi, P. Móra, F. Szöll\H{o}si
and M. \foreignlanguage{british}{Weiner}, \textit{J. Phys. A: Math.
Theor.} \textbf{42}, 245305 (2009).

\bibitem{mcnulty2} D. McNulty and S. Weigert, \textit{J. Phys. A:
Math. Theor.} \textbf{45}, 102001 (2012).

\bibitem{butterley07} P. Butterley and W. Hall, \textit{Phys. Lett.
A} \textbf{369}, 5 (2007).

\bibitem{brierley08} S. Brierley and S. Weigert, \textit{Phys. Rev.
A} \textbf{78}, 042312 (2008).

\bibitem{klappenecker04} A. Klappenecker and M. R{\" o}tteler, in \textit{Finite Fields and Applications}, vol. 2948 of \textit{Lecture Notes
in Computer Science}, pp. 137-144, Springer, Berlin, (2004).

\bibitem{wocjan05} P. Wocjan and T. Beth, \textit{Quantum Inf. Comput.}
\textbf{5}, 93-101 (2005).

\bibitem{wiesniak11} M. Wie\'{s}niak, T. Paterek and A. Zeilinger,
\textit{New J. Phys.} \textbf{13}, 053047 (2011).

\bibitem{pammer15} B. Pammer, \textit{Mutually Unbiased Quantum Bases: Existence, Entanglement, Information}, Wien (unpublished, 2015).

\bibitem{mcnulty_all} D. McNulty and S. Weigert, \textit{J. Phys.
A: Math. Theor.} \textbf{45}, 135307 (2012).

\bibitem{brierley_all} S. Brierley, S. Weigert and I. Bengtsson,
\textit{Quantum. Inf. Comput.} \textbf{10}, 803 (2010).

\bibitem{mcnulty+12d} D. McNulty and S. Weigert, \textit{Int. J.
Quant. Inf.} \textbf{10}, 1250056 (2012).\end{thebibliography}
\end{document}